\documentclass[11pt]{myclass}

\usepackage{amsmath,amssymb,epsfig,graphicx,graphics,color}
\usepackage{euscript}
\usepackage{picins}

\usepackage{macros}

\newcommand{\eps}{\varepsilon}

\renewcommand{\c}[1]{\ensuremath{\EuScript{#1}}}
\renewcommand{\b}[1]{\ensuremath{\mathbb{#1}}}

\newcommand{\IP}[2]{\ensuremath{ \langle #1 , #2 \rangle}}
\newcommand{\wid}{\omega}

\def\grid{\mathbb{G}}
\def\hcube{\mathbb{H}}
\def\excube{\mathbb{C}}
\def\naturals{\mathbb{N}}

\def\A{\EuScript{A}}
\def\G{\EuScript{G}}
\def\kernel{\EuScript{K}}
\def\wkernel{\EuScript{L}}
\def\nbr{\psi}
\def\inbr{\varphi}
\def\span{\mathop{\mathrm{span}}}
\def\ceil#1{\lceil#1\rceil}

\newcommand{\conv}[1]{\mathop{\mathrm{conv}}(#1)}
\newcommand{\opt}[2]{\kappa(#1,#2)}
\newcommand{\dual}[1]{{#1}^*}

\newcommand{\polyt}{\EuScript{P}}
\newcommand{\ball}{\mathbb{B}}
\newcommand{\intr}{\mathop{\mathrm{intr}}}
\newcommand{\fp}{\bar{p}}

\newcommand{\fq}{\bar{q}}
\newcommand{\fDelta}{\bar{\Delta}}
\newcommand{\fr}{\bar{r}}
\newcommand{\point}[3]{{#1}[#3,#2]}

\newcommand{\rotate}[2]{{#1}[#2]}
\newcommand\PPic[2]{
    \begin{minipage}{#1}%
        \hspace{-0.4cm}
        \makebox[0cm][l]{\includegraphics{#2}}
    \end{minipage}}

\newlength{\ppicwd}

\title{Stability of $\eps$-Kernels}
\author{Pankaj K. Agarwal%\footnote{} 
\and 
Jeff M. Phillips%\footnote{} 
\and 
Hai Yu}%$^\diamondsuit$}

\begin{document}
\begin{titlepage}
%\footnotetext{$^\dagger$Department of Computer Science, Duke University, Durham, NC 27708: \texttt{pankaj@cs.duke.edu}}
%\footnotetext{$^\ddagger$Department of Computer Science, Duke University, Durham, NC 27708: \texttt{jeffp@cs.duke.edu}}
%\footnotetext{$^\diamondsuit$Google, New York, NY: \texttt{fishhai@google.com}}
\maketitle

\begin{center} \today \end{center}

\begin{abstract}
Given a set $P$ of $n$ points in $\b{R}^d$, an $\eps$-kernel $K \subseteq P$ 
approximates the directional width of $P$ in every direction within a relative 
$(1-\eps)$ factor.  In this paper we study the stability of 
$\eps$-kernels under dynamic insertion and deletion of points to $P$ and by changing the approximation factor $\eps$.
In the first case, we say an algorithm for dynamically
maintaining a $\eps$-kernel is stable if at most $O(1)$ points
change in $K$ as one point is inserted or deleted from $P$.
We describe an algorithm to maintain an $\eps$-kernel of 
size $O(1/\eps^{(d-1)/2})$ in $O(1/\eps^{(d-1)/2} + \log n)$ 
time per update.  Not only does our algorithm maintain a
stable $\eps$-kernel, its update time is faster than any known
algorithm that maintains an $\eps$-kernel of size
$O(1/\eps^{(d-1)/2})$. 
Next, we show that if there is an $\eps$-kernel of $P$ of
size $\kappa$, which may be dramatically less than 
$O(1/\eps^{(d-1)/2})$, then there is 
an $(\eps/2)$-kernel of $P$ of size 
$O(\min\{ 1/\eps^{(d-1)/2},
	\kappa^{\lfloor d/2 \rfloor} \log^{d-2} (1/\eps)\})$.
Moreover, 
there exists a point set $P$ in $\reals^d$ and a parameter
$\eps > 0$ such that if every $\eps$-kernel of $P$ has size at
least $\kappa$, then any $(\eps/2)$-kernel of $P$ has size 
$\Omega(\kappa^{\lfloor d/2 \rfloor})$.
\footnote{%The work was primarily done when the second and third authors were at Duke University. 
Research supported by 
    subaward CIF-32 from NSF grant 0937060 to CRA, 
    by NSF under grants
    CNS-05-40347, CFF-06-35000, and DEB-04-25465, by ARO grants
    W911NF-04-1-0278 and W911NF-07-1-0376, by an NIH grant
    1P50-GM-08183-01, by a DOE grant OEG-P200A070505, and by a grant
    from the U.S.--Israel Binational Science Foundation.}
\end{abstract}
\end{titlepage}

%Given a set P of n points in |R^d, an eps-kernel K subset P approximates the directional width of P in every direction within a relative (1-eps) factor.  In this paper we study the stability of eps-kernels under dynamic insertion and deletion of points to P and by changing the approximation factor eps.
%In the first case, we say an algorithm for dynamically maintaining a eps-kernel is stable if at most O(1) points change in K as one point is inserted or deleted from P.  We describe an algorithm to maintain an eps-kernel of size O(1/eps^{(d-1)/2}) in O(1/eps^{(d-1)/2} + log n) time per update.  Not only does our algorithm maintain a stable eps-kernel, its update time is faster than any known algorithm that maintains an eps-kernel of size O(1/eps^{(d-1)/2}). 
%Next, we show that if there is an eps-kernel of P of
%size k, which may be dramatically less than 
%O(1/eps^{(d-1)/2}), then there is an (eps/2)-kernel of P of size 
%  O(min { 1/eps^{(d-1)/2}, k^{floor(d/2)} log^{d-2} (1/eps) }).
%Moreover, there exists a point set P in |R^d and a parameter
%eps > 0 such that if every eps-kernel of P has size at least k, then any (eps/2)-kernel of P has size 
% Omega(k^{floor(d/2)}).

%%%%%%%%%%%%%%%%%%%%%%%%%%%%%%%%%%%%%%%%%%%%%%%%%%
%%%%%%%%%%%%%%%%%%%%%%%%%%%%%%%%%%%%%%%%%%%%%%%%%%
\section{Introduction}
\label{sec:intro}

With recent advances in sensing technology, massive geospatial
data sets are being acquired at an unprecedented rate in many application
areas, including GIS, sensor networks, robotics, and spatial databases. Realizing the full potential of these data 
sets requires developing scalable algorithms for 
analyzing and querying them.  Among many interesting
algorithmic developments to meet this challenge,  there is an
extensive amount of work on computing a ``small summary'' of large 
data sets that preserves certain desired properties of the input data 
and on obtaining a good trade-off between the quality of the summary
and its size.  A coreset is one example of such approximate summaries.
Specifically, for an input set $P$ and a function $f$, 
a \emph{coreset} $C \subseteq P$ is a subset of $P$ (with respect to $f$)
with the property that $f(C)$ approximates $f(P)$. If a small-size coreset
$C$ can be computed quickly (much faster than computing $f(P)$), then 
one can compute an approximate value of $f(P)$ by first
computing $C$ and then computing $f(C)$. This 
coreset-based approach has been successfully used in a wide range of 
geometric optimization problems over the last decade.  See~\cite{AHV07} for a survey.

%%%%%%%%%%%%%%%%%%%%%%%%%%%%%%%%%%%%%%%%%%
\paragraph{$\eps$-kernels.} Agarwal~\etal~\cite{AHV04} introduced the notion
of $\eps$-kernels and proved that it is a coreset for many functions.
For any direction $u \in \b{S}^{d-1}$, let 
$P[u] = \arg \max_{p \in P} \IP{p}{u}$ be the extreme point in $P$ 
along $u$;  $\wid(P,u) = \IP{P[u] - P[-u]}{u}$ is called 
the \emph{directional width} of $P$ in direction $u$.  
For a given $\eps>0$, $K \subset P \subset \b{R}^d$ is called 
an \emph{$\eps$-kernel} of $P$ if 
$$
\IP{P[u] - K[u]}{u} \leq \eps \wid(P,u)
$$
for all directions $u \in \b{S}^{d-1}$.\footnote{%
This is a slightly stronger version of the definition than defined in 
\cite{AHV04} and an $\eps$-kernel $K$ gives a relative $(1+2\eps)$-approximation of $\wid(P,u)$ for all $u \in \b{S}^{d-1}$ (i.e. $\wid(K, u) \leq \wid(P,u) \leq (1+2\eps) \wid(K,u)$).}
%It has been shown in~\cite{AHV04,AHV07,Cha06,Cha08} that
%$\eps$-kernels lead to efficient approximation algorithms for
%a wide range of problems.
For simplicity, we assume $\eps \in (0, 1)$, because for $\eps \geq 1$, 
one can choose a constant number of points to form an $\eps$-kernel, and we assume $d$ is constant.  
By definition, if $X$ is an $\eps$-kernel of $P$ and $K$ is a $\delta$-kernel
of $X$, then $K$ is a $(\delta+\eps)$-kernel of $P$. 
%For a set $A$ (not necessarily a subset of $P$) of points, we call $K$ an 
%\emph{$\eps$-kenel with respect to $A$} if 
%$\IP{P[u] - K[u]}{u} \leq \eps \wid(A,u)$ for all $u\in \sphere^{d-1}$.

Agarwal~{et~al.}~\cite{AHV04} showed that there exists
an $\eps$-kernel of size $O(1/\eps^{(d-1)/2})$ and it can be
computed in time $O(n+1/\eps^{3d/2})$. The running time was improved 
by Chan~\cite{Cha06} to $O(n + 1 / \eps^{d-3/2})$ (see
also~\cite{YAPV04}).
In a number of applications, the input point set is being updated
periodically, so algorithms have also been developed to
maintain $\eps$-kernels dynamically.
%update coresets on the fly~\cite{BCEG04,Cha08,AY07,HS08,STZ04}.  
Agarwal~\etal~\cite{AHV04} had described a data structure to maintain
 an $\eps$-kernel of size $O(1/\eps^{(d-1)/2})$ 
in $(\log (n)/\eps)^{O(d)}$ time per update.
The update time was recently improved by Chan~\cite{Cha08} to
$O((1/\eps^{(d-1)/2})\log n + 1/\eps^{d-3/2})$. His approach can 
also maintain an $\eps$-kernel of size $O((1/\eps^d)\log n)$ with 
update time $O(\log n)$.  If only insertions are allowed (e.g.\ in a streaming model),
the size of the data structure can be improved to 
$O(1/\eps^{(d-1)/2})$~\cite{AY07,ZZ08}.

In this paper we study two problems related to the \emph{stability} of 
$\eps$-kernels: how $\eps$-kernels change as we update the input set or vary
the value of $\eps$.

%%%%%%%%%%%%%%%%%%%%%
\paragraph{Dynamic stability.}
Since the aforementioned dynamic algorithms for maintaining an
$\eps$-kernel focus on minimizing the size of the
kernel, changing a single point in the input set $P$ may 
drastically change the resulting kernel.  This is particularly 
undesirable when the resulting kernel is used to build 
a dynamic data structure for maintaining another information.
For example, kinetic data structures (KDS) based on coresets have 
been proposed to maintain various
extent measures of a set of moving points~\cite{AHV07}. If an
insertion or deletion of an object changes the entire summary,
then one has to reconstruct the entire KDS instead of locally
updating it. In fact, many other dynamic data structures for maintaining
geometric summaries also suffer from this undesirable 
property~\cite{BCEG04,HS08,STZ04}.

We call an $\eps$-kernel
\emph{$s$-stable} if the insertion or deletion of a point 
causes the $\eps$-kernel to change by at most $s$ points.
For brevity, if $s=O(1)$, we call the $\eps$-kernel to be
\emph{stable}. Chan's dynamic algorithm can be adapted to maintain
a stable $\eps$-kernel of size $O((1/\eps^{d-1})\log n)$; see 
Lemma~\ref{lemma:chan} below.
An interesting question is whether there is an efficient algorithm 
for maintaining a stable $\eps$-kernel of size
$O(1/\eps^{(d-1)/2})$, as points are being inserted or
deleted.  Maintaining a stable $\eps$-kernel dynamically 
is difficult for two main reasons.  
First, for an input set $P$, many algorithms compute $\eps$-kernels in 
two or more steps. They first construct a large $\eps$-kernel 
$K^\prime$ (e.g. see~\cite{AHV04,Cha08}), and then use a more 
expensive algorithm to create a small 
$\eps$-kernel of $K^\prime$.  However, if the first algorithm is unstable, then $K^\prime$ may change completely each 
time $P$ is updated.  
Second, all of the known $\eps$-kernel algorithms rely on first finding 
a ``rough shape'' of the input set $P$ (e.g., finding a small
box that contains $P$), estimating its fatness~\cite{BH01}.  
This rough approximation is used crucially in the computation
of the $\eps$-kernel.  However, this shape is itself very unstable 
under insertions or deletions to $P$.  Overcoming these
difficulties, we prove the following in
Section~\ref{sec:dynamic}:

\begin{theorem}
\label{theo:dynamic}
Given a parameter $0 \leq \eps \leq 1$, 
a stable $\eps$-kernel of size $O(1/\eps^{(d-1)/2})$ of a set of $n$ points in $\b{R}^d$ can be maintained under insertions and deletions in $O(1/\eps^{(d-1)/2} + \log n)$ time.
\end{theorem}

Note that the update time of maintaining an
$\eps$-kernel of size $O(1/\eps^{(d-1)/2})$ is better than
that in~\cite{Cha08}.
%%%%%%%%%%%%%%%%%%%%%
\paragraph{Approximation stability.}
If the size of an $\eps$-kernel $K$ is $O(1/\eps^{(d-1)/2})$,
 then decreasing $\eps$ changes $K$ quite predictably.
However, this is the worst-case bound, and
it is possible that the size of $K$ may be quite small, e.g.,
$O(1)$, or in general much smaller than the $1/\eps^{(d-1)/2}$ maximum
(efficient algorithms are known for computing
$\eps$-kernels of near-optimal size~\cite{AHV07}).  Then 
how much can the size increase as we reduce the allowable
error from $\eps$ to $\eps/2$?  
For any $\eps > 0$, let $\opt{P}{\eps}$ denote the minimum size of an  
$\eps$-kernel of $P$. 
Unlike many shape simplification problems, in which the size of
simplification can change drastically as we reduce the value
of $\eps$, we show (Section~\ref{sec:approx}) that 
this does not happen for $\eps$-kernels and that
$\opt{P}{\eps/2}$ can be expressed in terms of $\opt{P}{\eps}$. 

\begin{theorem}
\label{theo:approx}
    For any point set $P$ and for any $\eps>0$, 
    $$\opt{P}{\eps/2} = O(
	\min \{ \opt{P}{\eps}^{\lfloor d/2 \rfloor} \log
		^{d-2} (1/\eps), 1/\eps^{(d-1)/2} \}).$$
Moreover, there exist a point set $P$ and some $\eps>0$ such that 
   $\opt{P}{\eps/2} =\Omega (\opt{P}{\eps}^{\lfloor d/2 \rfloor})$.
\end{theorem}

%%%%%%%%%%%%%%%%%%%%%%%%%%%%%%%%%%%%%%%%%%
%%%%%%%%%%%%%%%%%%%%%%%%%%%%%%%%%%%%%%%%%%
%%%%%%%%%%%%%%%%%%%%%%%%%%%%%%%%%%%%%%%%%%
\section{Dynamic Stability}
\label{sec:dynamic}

%%%%%%%%%%% OUTLINE %%%%%%%%%%%%%%%%%
In this section we describe an algorithm that proves 
Theorem~\ref{theo:dynamic}. The algorithm is composed of a sequence of 
modules, each with certain property.  
%
%We first state that Chan's dynamic coreset algorithm~\cite{Cha08} can be made stable (see appendix):
%
%\begin{lemma}
%\label{lemma:chan}
%For any $0 < \eps <1$, an $\eps$-kernel $\kernel$ of $P$ of size $O((1/ \eps^{d-1}) \log n)$ can be maintained in $O(\log  n)$ time with $O(1)$ changes to $\kernel$ per update. 
%%The amortized number of changes in $\kernel$ at each update is $O(1)$.
%\end{lemma}
%
%
We first define the notion of anchor points and fatness of a 
point set and describe two algorithms for maintaining 
stable $\eps$-kernels with respect to a fixed anchor:
one of them maintains a kernel of size $O(1/\eps^{d-1})$ and the other 
of size $O(1/\eps^{(d-1)/2})$; the former has smaller update time.
Next, we briefly sketch how Chan's algorithm~\cite{Cha08} can be adapted to maintain a stable $\eps$-kernel of size $O(1/\eps^{(d-1)}\log n)$. 
Then we describe the algorithm for updating anchor points and 
maintaining a stable kernel as the anchors change. Finally, we put these 
modules together to obtain the final algorithm.
We make the following simple observation, which will be crucial for combining different modules.

\begin{lemma}[\textsf{Composition Lemma}]
If $K$ is an $s$-stable $\eps$-kernel of $P$ and $K^\prime$ is an $s^\prime$-stable $\eps^\prime$-kernel of $K$, then $K^\prime$ is an $(s \cdot s^\prime)$-stable $(\eps + \eps^\prime)$-kernel of $P$.
\label{lem:chain}
\end{lemma}

\paragraph{Anchors and fatness of a point set.}
%\label{ssec:make-fat}
%%%%%% fatness and anchor points
%\paragraph{Anchor points and transformations.}
We call a point set $P$ \emph{$\beta$-fat} if 
$$
\max_{u,v \in \b{S}^{d-1}} \wid(P,u)/\wid(P,v) \leq \beta.
$$
If $\beta$ is a constant, we sometimes just say that $P$ is \emph{fat}.  
An arbitrary point set $P$ can be made fat by applying an affine 
transform: we first choose a set of $d+1$ \emph{anchor points} 
$A = \{a_0, a_1, \ldots, a_d\}$ using the following procedure of 
Barequet and Har-Peled~\cite{BH01}.  Choose $a_0$ arbitrarily.  
Let $a_1$ be the farthest point from $a_0$.  Then inductively, 
let $a_i$ be the farthest point from the flat 
$\span (a_0, \ldots, a_{i-1})$.  (See Figure \ref{fig:transformTA}.)
The anchor points $A$ define a bounding box $I_A$ with center at $a_0$ 
and orthogonal directions defined by vectors from the flat 
$\span(a_0, \ldots, a_{i-1})$ to $a_i$.  The extents of $I_A$ in each 
orthogonal direction is defined by placing each $a_i$ on a bounding 
face and extending $I_A$ the same distance from $a_0$ in the 
opposite direction.  
Next we perform an affine transform $T_A$ on $P$ such that the vector from the flat $\span(a_0, \ldots, a_{i-1})$ to $a_i$ is equal to $e_i$, where $e_0 = (0,\ldots,0), e_1 = (1,0,\ldots,0), \ldots,$ $e_d = (0, \ldots, 0,1)$.  
This ensures that $T_A(P) \subseteq T_A(I_A)  = [-1,1]^d$.  The next
lemma shows that $T_A(P)$ is fat.

\begin{lemma}
\label{lemma:fat}
For all $u \in \b{S}^{d-1}$ and for $\beta_d \leq 2^d d^{3/2}  d!$,
\begin{equation}
\wid(T_A(A),u) \leq \wid(T_A(P),u) \leq \wid(T_A(I_A), u) \leq \beta_d \cdot \wid(T_A(A),u).
\label{eq:fat}
\end{equation}
\end{lemma}

\begin{proof}
The first two inequalities follow by $A \subset P \subset I_A$.  We can upper bound $\max_u \wid(T_A(I_A),u) = \wid([-1,1]^d,u) \leq \sqrt{d}$.
The volume of the convex hull $\textsf{conv}(T_A(A))$ is $1/d!$ since it is a $d$-simplex and $\IP{T_A(a_0)-T_A(a_i)}{e_i} = 1$ for each direction $e_i$.  
We can then scale $\textsf{conv}(T_A(A))$ by a factor $1/2$ (shrinking the volume by factor $1/2^d$) so it fits in $[0,1]^d$.  Now we can apply a lemma from \cite{HPbook} that the minimum width of a convex shape that is contained in $[0,1]^d$ is at least $1/d$ times its $d$-dimensional volume, which is $1/(2^d d!)$.  
The fatness of $T_A(P)$ follows from the fatness of $T_A(I_A)$.
\end{proof}

%\noindent\textbf{Remark.} It is easily seen that it suffices $a_i$ to be 
%an approximate farthest neighbor of $\span(a_0,\ldots,a_{i-1})$ in the% 
%above construction. The constant $\beta_d$ depends on the 
%approximation factor.

\vspace{-.1in}

\begin{figure}[htb]
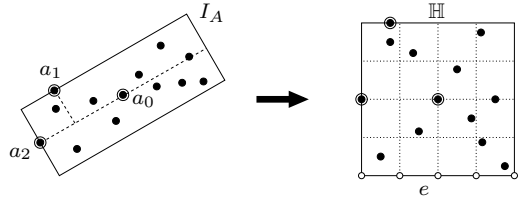

%\center{\includegraphics{transform.pdf}
\begin{center}
\input figs/transform
\caption{ Anchor points $A=\{a_0,a_1,a_2\}$, rectangle $I_A$, and 
transform $T_A$ applied to $P$;  square $\hcube$,
two-dimensional grid $\grid$, and one-dimensional grid
$\grid_e$ on the edge $e$ of $\hcube$.}
\label{fig:transformTA}
\end{center}
\end{figure}

Agarwal~\etal~\cite{AHV04} show if $K$ is an $\eps$-kernel of $P$, then 
$T(K)$ is an $\eps$-kernel of $T(P)$ for any affine transform $T$, which 
implies that one can compute an $\eps$-kernel of $T(P)$. We will need 
the following generalization of the definition of $\eps$-kernel.
For two points sets $P$ and $Q$, a subset $K \subseteq P$ is called
an \emph{$\eps$-kernel of $P$ with respect to $Q$} if
$\IP{P[u] - K[u]}{u} \leq \eps \wid(Q,u)$ for all $u \in \sphere^{d-1}$.
%(\ref{eq:fat}) holds, then we say \emph{$P$ is fat with respect to $A$}.  

%%%%%%%%%%%%%%%%%%%%%%%%%%%%%%%%%%
\paragraph{Stable $\eps$-kernels for a fixed anchor.}
%\label{ssec:fat}

Let $A$ be a set of anchor points of $P$, as described above. We
describe algorithms for maintaining stable $\eps$-kernels (with respect
to $A$) under the assumption that $A$ remains a set of  anchor points of
$P$, i.e., $A \subseteq P \subset I_A$, as $P$ is being 
updated by inserting
and deleting points. In view of the above discussion, without loss of generality,
we assume $I_A = [-1,+1]^d$ and denote it by $\hcube$. As for the static
case~\cite{AHV04,Cha06}, we first describe a simpler algorithm that
maintains a stable $\eps$-kernel of size $O(1/\eps^{d-1})$, and then a 
more involved one that maintains a stable $\eps$-kernel of size 
$O(1/\eps^{(d-1)/2})$. 

Set $\delta=\eps/\sqrt{d}$ and draw a $d$-dimensional grid $\grid$ inside
$\hcube$ of size $\delta$, i.e., the side-length of each grid cell is
at most $\delta$; $\grid$ has $O(1/\delta^d)$ cells. For each grid
cell $\tau$, let $P_\tau = P \cap \tau$. For a point $x \in \hcube$ lying 
in a grid cell $\tau$, 
let $\hat{x}$ be the vertex of $\tau$ nearest to the origin; we can
view $x$ being \emph{snapped} to the vertex $\hat{x}$.
For each facet $f$ of $\hcube$,
$\grid$ induces a $(d-1)$-dimensional grid $\grid_f$ on $f$; $\grid$ contains
a \emph{column} of cells for each cell in $\grid_f$. For 
each cell $\Delta \in \grid_f$, we choose (at most) one point of $P$ 
as follows:
let $\tau$ be the nonempty grid cell in the column of $\grid$ 
corresponding to $\Delta$ that is closest to $f$. We choose an 
arbitrary point from $P_\tau$;
if there is no nonempty cell in the column, no point is chosen. Let
$L_f$ be the set of chosen points. Set $\wkernel = 
\bigcup_{f \in\hcube} L_f$. Agarwal~\etal~\cite{AHV04} proved that 
$\wkernel$ is an $\eps$-kernel of $P$.  Insertion or deletion of a point
in $P$ affects at most one point in $L_f$, and it can be updated in 
$O(\log (1/\eps))$ time. Hence, we obtain the following:

\begin{lemma}
Let $P$ be a set of $n$ points in $\reals^d$, let 
$A \subseteq P$ be a set of anchor points of $P$, and let $0 < \eps < 1$
be a parameter.  $P$ can be preprocessed in $O(n + 1/\eps^{d-1})$ time,
so that a (2d)-stable $\eps$-kernel of $P$ with respect to $A$ of 
size $O(1/\eps^{d-1})$ can be maintained in $O(\log 1/\eps)$ time 
per update provided that $A$ remains an anchor set of $P$.
\label{lem:weak-fixed}
\end{lemma}

%%%%%%%%%%%%%%%%%%%%%%%%%%%%%%%%%%%%%%%%%%%%%%%%
Agarwal~\etal~\cite{AHV04} and Chan~\cite{Cha06} have described algorithms for 
computing an $\eps$-kernel of size $O(1/\eps^{(d-1)/2})$. We adapt Chan's 
algorithm to maintain a stable $\eps$-kernel with respect to a 
fixed anchor $A$. We begin by mentioning a result of Chan that lies 
at the heart of his algorithm.

\begin{lemma}[Chan~\cite{Cha06}]
\label{lemma:discrete}
Let $E \in \naturals$, $E^\tau \le F \le E$ for some $0 < \tau < 1$, and
$P \subseteq [0:E]^{d-1}\times\reals$ a set of at most $n$ %$E^{d-1}$ 
points. For all grid points $b \in [0:F]^{d-1}\times\reals$, the nearest
neighbors of each $b$ in $P$ can be computed in time $O(n + E^{d-2}F)$.
\end{lemma}

We now set $\gamma = \sqrt{\eps}/c$ for a constant $c > 1$ to be used in a much sparser grid than with $\delta$. 
Let $\excube =[-2,+2]^d$ and $f$ be a facet of $\excube$. 
We draw a $(d-1)$-dimensional grid on $f$ of size $\gamma$. 
Assuming $f$ lies on the plane $x_d=-2$, we choose a set $B_f = \{ (i_1\gamma,\ldots,i_{d-1}\gamma,-2) \in \b{Z}^d \mid -\ceil{2/\gamma} \le i_1,\ldots,i_{d-1} \le \ceil{2/\gamma}\}$ of grid points.
For a subset $X \subseteq P$ and a point $b$, we define $\nbr(X,b) = \arg\min_{x\in X} \|\hat{x}-b\|$, i.e., the point in $X$ such that the snapped point is nearest to $b$. 
For a set $R$, $\nbr(X,R) = \{ \nbr(X,r) \mid r \in R\}$. 
There is a one to one mapping between the faces of $\excube$ and $\hcube$, so we also use $f$ to denote the corresponding facet of $\hcube$. 
Let $L_f$ be the set of points chosen in the previous algorithm corresponding to facet $f$ of $\hcube$ for computing an $(\eps/2)$-kernel of $P$. 
Set $G_f = \nbr(L_f,B_f)$.  
Chan showed that $\G = \bigcup_{f \in \excube} G_f$ is an $(\eps/2)$-kernel of $\wkernel$ and thus an $\eps$-kernel of $P$. 
Scaling $\grid$ and $B_f$ appropriately and using Lemma~\ref{lemma:discrete}, $G_f$ can be computed in $O(n + 1/\eps^{d-3/2})$ time. Hence, $\G$ can be computed in $O(n+1/\eps^{d-3/2})$ time.

Note that $\nbr(L_f,b)$ can be the same for many points $b\in B_f$, so 
insertion or deletion of a point in $P$ (and thus in $L_f$) may change
$G_f$ significantly, thereby making $\G$ unstable. We circumvent this problem by introducing two new ideas. First, $\nbr(L_f,B_f)$ is computed 
in two 
stages, and second it is computed in an iterative manner. We describe the 
construction and the update algorithm for $f$; the same algorithm is repeated
for all facets. 

We partition $\hcube$ into $O(1/\gamma^{d-1})$ boxes: for 
$J =\langle i_1.\ldots,i_{d-1}\rangle \in [-1/\gamma,1/\gamma]^{d-1} \cap \b{Z}^{d-1}$, 
we define
$\hcube_J = [i_1\gamma,(i_1+1)\gamma]\times \cdots \times [i_{d-1}\gamma,
(i_{d-1}+1)\gamma] \times [-1,+1]$. We maintain a subset
$X \subseteq L_f$. Initially, we set $X=L_f$. Set $X_J = X \cap \hcube_J$.
We define a total order on the points of $B_f$. Initially, we
sort $B_f$ in lexicographic order, but the ordering will change as 
insertions and deletions are performed on $P$.
Let $\langle b_1, \ldots, b_u\rangle$ be the current ordering of $B_f$. We 
define a map $\inbr: B_f \rightarrow L_f$ as follows. Suppose
$\inbr(b_1), \ldots, \inbr(b_{i-1})$ have been defined. Let 
$J_i = \arg\min_J \|\hat{\nbr}(X_J,b_i)-b_i\|$; here 
$\hat{\psi}(\cdot)$ denotes the snapped point of $\psi(\cdot)$. We set
$\inbr(b_i) = \nbr(X_{J_i},b_i)$. We delete $\inbr(b_i)$ from $X$ (and 
from $X_{J_i}$) and recompute $\hat{\nbr}(X_{J_i},B_f)$. Set $K_f = \{\inbr(b) 
\mid b \in B_f\}$ and $\kernel = \bigcup_f K_f$.
Computing $J_i$ and $\inbr(b_i)$ takes 
$O(1/\eps^{(d-1)/2})$ time, and, by Lemma~\ref{lemma:discrete}, 
$\nbr(X_{J_i},B_f)$ can be computed in $O(|X_J| + 1/\gamma^{d-2} \cdot 1/\gamma) = O(1/\eps^{(d-1)/2})$ time.

It can be proved that the map $\inbr$ and the set $K_f$ satisfy the following properties:
\begin{itemize} \denselist
\item[(P1)] $\inbr(b_i) \ne \inbr(b_j)$ for $i \ne j$,
\item[(P2)] $\inbr(b_i) = \nbr(L_f\setminus\{\inbr(b_j)\mid j < i\}, b_i)$,
\item[(P3)] $\displaystyle K_f \supseteq \nbr(L_f,B_f)$.
\end{itemize}
Indeed, (P1) and (P2) follow from the construction, and (P3) follows from
(P2). (P3) immediately implies that $\kernel$ is an $\eps$-kernel of $P$.
Next, we describe the procedures for updating $K_f$ when $L_f$
changes.  These procedures maintain (P1)--(P3), 
thereby ensuring that the algorithm maintains an $\eps$-kernel. 

\textit{Inserting a point.} 
Suppose a point $p$ is inserted into $L_f$. 
%Let $\tau_p$ be the grid cell of $\grid$ containing $p$. 
%If $c_p$ is already nonempty or if $c_p$ is not the closest nonempty cell
% to $f$ in its column, then the set $L_f$ does not change by insertion 
%of $p$, so we stop. Otherwise, 
We add $p$ to $X$. Suppose $p \in \hcube_J$. We recompute $\nbr(X_J,B_f)$. Next, we update $\inbr(\cdot)$ and 
$\kernel$ as  follows. We maintain a point $\xi \in L_f$. Initially, 
$\xi$ is set to $p$.  Suppose we have processed 
$b_1, \ldots, b_{i-1}$. Let $\eta \in L_f$ be the 
current $\inbr(b_i)$. If $\|\hat{\xi}-b_i\| \le \|\hat{\eta}-b_i\|$, then we 
swap $\xi$ and $\inbr(b_i)$, otherwise neither $\xi$ nor $\inbr(b_i)$ is updated.
We then process $b_{i+1}$.  After processing all points of $B_f$ if $\xi=p$, 
i.e., no $\inbr(b_i)$ is updated, we stop. Otherwise, we add $p$ to $K_f$ and 
delete $\xi$ from $K_f$. 
The insertion procedure makes at most two changes in $K_f$, and it can be 
verified that (P1)-(P3) are maintained.

\textit{Deleting a point.} 
Suppose $p$ is deleted from $L_f$.  Suppose $p \in \hcube_J$. 
If $p \not\in K_f$, then $p \in X$. 
We delete $p$ from $X$ and $X_J$ and recompute $\nbr(X_J,B)$. 
If $p \in K_f$, i.e., there is a $b_i \in B$ with $p=\inbr(b_i)$, then $p \not\in X$.
% and we do not recompute $\nbr(X_J,B)$. Instead,
We delete $p$ from $K_f$ and $\kernel$, recompute $\inbr(b_i)$, and 
add the new $\inbr(b_i)$ to $K_f$.  Let $\inbr(b_i) \in \hcube_J$; we remove $\inbr(b_i)$ from $X_J$ and recompute $\nbr(X_J,B_f)$.
We modify the ordering of $B_f$ by moving $b_i$ from its current position to the end. This is the only place where the ordering of $B_f$ is modified. 
Since $b_i$ is now the last point in the ordering of $B_f$, the new $\inbr(b_i)$ does not affect any other $\inbr(b_j)$. 
The deletion procedure also makes at most two changes in $K_f$ and maintains (P1)--(P3).  

Finally, insertion or deletion of a point in $P$ causes at most one insertion plus one deletion in $L_f$, therefore we can conclude the following:

\begin{lemma}
\label{lem:strong-fixed}
Let $P$ be a set of $n$ points in $\reals^d$, $A$ a set of 
anchor points of $P$, and $0 < \eps < 1$ a parameter. 
$P$ can be preprocessed in $O(n+1/\eps^{d-1})$ time into a 
data structure so that a stable $\eps$-kernel of $P$ with respect to $A$
of size $O(1/\eps^{(d-1)/2})$ can be maintained in 
$O(1/\eps^{(d-1)/2})$ time under insertion and deletion, 
provided that $A$ remains an anchor set of $P$.
\end{lemma}

%%%%%%%%%%%%%%%%%%%%%%%%%%%%%%%%%%%%%%%%%%%%%
\paragraph{Stabilizing Chan's dynamic algorithm.}
We now briefly describe how Chan's \cite{Cha08} dynamic $\eps$-kernel algorithm can be adapted so that it maintains a stable $\eps$-kernel of size $O((1/\eps^{d-1}) \log n)$. 
He bypasses the need of fixed anchors by partitioning $P$ into $h=O(\log n)$ layers $\langle P_1, \ldots, P_h\rangle$, where $P_1$ is the inner-most layer and $P_h$ is the outer-most layer,
$|P_i| \geq \gamma \sum_{j=i+1}^h |P_j|$ for a constant $\gamma >1$, and $|P_h| = 1/\eps^{d-1}$. 
$P_1$ is constructed first, and then the rest of the layers are constructed recursively with the remaining points.
For each set $P_i$ (for $i < h$) there exists a set of points $A_i$ which serve as anchor points for $P_i$ in the sense that 
they define a bounding box $I_{A_i}$ (i.e., $P_i \subset I_{A_i}$).  Furthermore, for all $u \in \b{S}^{d-1}$ we have $\wid(A_i,u) \leq \wid(P_{i+1},u)$, and this remains true for $\alpha |P_i|$ insertions or deletions to $P$ for a constant $0 \leq \alpha \leq 1$.
After $\alpha|P_i|$ updates in $P$, the layers $P_i, P_{i+1}, \ldots, P_h$ are reconstructed. %the update conditions are actually slightly more complicated and are described in the proof of Lemma \ref{lemma:chan} in the proof. 
Also at this point layers $P_j, P_{j+1}, \ldots, P_{i-1}$ will need to be reconstructed if layer $P_j$ is scheduled to be reconstructed in fewer than $\alpha |P_i|$ updates.  
We set $K_h = P_h$, and for $i < h$, an $\eps$-kernel $K_i$ of $P_i$ with respect to $A_i$ is maintained using Lemma~\ref{lem:weak-fixed}.
The set $\kernel = \bigcup_{i=1}^h K_i$ is an $\eps$-kernel of $P$; 
$|\kernel| = O((1/\eps^{d-1})\log n)$.

When a new point $p$ is inserted into $P$, it is added to the outermost layer $P_i$, i.e., $i$ is the largest such value, such that 
$p \in I_{A_i}$. If a point is inserted into or deleted from $P_i$,
we update $K_i$ using Lemma \ref{lem:weak-fixed}. The update time follows
from the following lemma. 

\begin{lemma}\label{lemma:chan}
For any $0 < \eps <1$, an $\eps$-kernel $\kernel$ of $P$ of size $O((1/ \eps^{d-1}) \log n)$ can be maintained in $O(\log  n)$ time, and the number of changes in $\kernel$ at each update is $O(1)$.
\end{lemma}

\begin{proof}
Recall that if the insertion or deletion of a point does not require reconstruction of any layers, the update time is $O(\log n)$ and by Lemma~\ref{lem:weak-fixed}, only $O(1)$ changes occur in $\kernel$. 
We start by bounding the amortized time spent in reconstructing layers and their kernels.  

The kernel $K_i$ of $P_i$ is rebuilt if at least $\alpha|P_i|$ updates have taken place since the last reconstruction, if $K_l$ needs to be rebuilt for some $l < i$, or 
%(a detail skipped in the main body text) 
if some $K_j$ is rebuilt for $j > i$ and $K_i$ is scheduled to be rebuilt in fewer than $\alpha |P_j|$ updates to $P$.  These conditions imply that the $i$th layer is rebuilt after every $\alpha |P_h| 2^j$ updates where $j$ is the smallest integer such that $|P_i| \leq |P_h| 2^j$.  
The third condition acts to coordinate the rebuilding of the layers so that the $i$th layer is not rebuilt after fewer than $\alpha|P_i|/2$ updates since its last rebuild.  

$K_h$ is rebuilt after $\alpha/\eps^{d-1}$ updates.  And $|P_i| \leq \gamma^{h-i} |P_h|$.  Since the entire system is rebuilt after $\Theta(|P_1|) = \Theta(n)$ updates, we call this interval a round.  We can bound the updates to $\kernel$ in a round by charging $O(1/\eps^{d-1})$ each time a $K_i$ is rebuilt, which occurs at most $\Theta(|P_1|) 2/\alpha |P_i| = \Theta(n/|P_i|)$ times in a round.
$$
\sum_{i=1}^h \frac{O(n)}{\Omega(|P_i|)} \cdot O(1/\eps^{d-1})
\leq
\sum_{i=1}^h \frac{O(n) \cdot O(1/\eps^{d-1})}{\Omega(\gamma^{h-i} |P_h|)}
=
\sum_{i=1}^h \frac{O(n)}{\Omega(\gamma^i)}
= 
O(n).
$$
Thus there are $O(n)$ updates to the $\eps$-kernel $\kernel$ for every $\Theta(n)$ updates to $P$.  Thus, in an amortized sense, for each update to $P$ there are $O(1)$ updates to $\kernel$.  

This process can be 
de-amortized by adapting the standard techniques for de-amortizing the 
update time of a dynamic data structure~\cite{Ove83}. 
If a kernel $K_i$ is valid for $k$ insertions or deletions to $P$, then we start construction on the next kernel $K_i$ after $k/2$ insertions or deletions have taken place since the last time $K_i$ was rebuilt.  All insertions can be put in a queue and added to $K$ by the time $k/2$ steps have transpired.  All deletions from old $K_i$ to new $K_i$ are then queued and removed from $K$ before another $k/2$ insertions or deletions.  This can be done by performing $O(1)$ queued insertions or deletions from $K$ each insertion or deletion from $P$.  
\end{proof}

%%%%%%%%%%%%%%%%%%%%%%%%%%%%%%%%%%%%%%%%%%%%%
\paragraph{Updating anchors.}

We now describe the algorithm for maintaining a stable $\eps$-kernel when
anchors of $P$ are no longer fixed and need to be updated dynamically. 
Roughly speaking,
we divide $P$ into \emph{inner} and \emph{outer} subsets of points. 
The outer subset acts as a \emph{shield} so that a stable kernel of the inner
subset with respect to a fixed anchor can be maintained using 
Lemma~\ref{lem:weak-fixed} or~\ref{lem:strong-fixed}. When the 
outer subset can no longer act as a shield, we reconstruct the inner
and outer sets and start the algorithm again. We refer to the duration 
between two consecutive reconstruction steps as an \emph{epoch}. 
The algorithm maintains a stable kernel within each epoch, and 
the amortized number of changes in the kernel because of 
reconstruction at the beginning of a new epoch will be $O(1)$. As above, we 
use the same de-amortization technique to make the $\eps$-kernel stable
across epochs. We now describe the algorithm in detail.

In the beginning of each epoch, we perform the following preprocessing.
Set $\alpha=1/10$ and compute a $\alpha$-kernel $\wkernel$ of $P$ 
of size $O(\log n)$ using Chan's dynamic algorithm; we do not need the stable 
version of his algorithm described above.  $\wkernel$ can be updated
in $O(\log n)$ time per insertion/deletion.
We choose a parameter $m$, which is set to $1/\eps^{d-1}$ or 
$1/\eps^{(d-1)/2}$. We create the outer subset of $P$ by peeling off 
$m$ ``layers'' of anchor points $A_1, \ldots, A_m$.
Initially, we set $P_0 = P$. Suppose we have constructed 
$A_0,\ldots,A_{i-1}$. Set $P_{i-1}=P \setminus \bigcup_{j=1}^{i-1} A_j$,
and $\wkernel$ is an $\alpha$-kernel of $P_{i-1}$.
Next, we construct the anchor set $A_i$ of $\wkernel$ 
as described earlier in this section. We set $P_i = P_{i-1}\setminus A_i$
and update $\wkernel$ so that it is an $\alpha$-kernel of $P_i$. Let 
$\A = \bigcup_i A_i$, $A = A_m$, and $P_I = P \setminus \A$. Let
$\hcube= (1+\alpha)I_A$. By construction $P_I \subset \hcube$.  $\A$ forms the outer subset and acts as a shield for $P_I$, which is the inner subset. 
Set $\delta = \eps/(2(1+\alpha)(\beta_d)^2)$, where $\beta_d$ is the constant in Lemma~\ref{lemma:fat}.

If $m=1/\eps^{d-1}$ (resp.\ $1/\eps^{(d-1)/2}$), we maintain a stable
$\delta$-kernel $\kernel_I$ of $P_I$ with respect to $A$ of size $O(m)$
using Lemma~\ref{lem:weak-fixed} (resp.\ Lemma~\ref{lem:strong-fixed}).
Set 
$\kernel = \kernel_I \cup \A$; $|\kernel|=O(m)$. We prove below that $\kernel$ is an 
$\eps$-kernel of $P$.
Let $p$ be a point that is inserted into or deleted from $P$.
If $p \in \hcube$, then we update  $\kernel_I$ using 
Lemma \ref{lem:weak-fixed} or \ref{lem:strong-fixed}.  On the other hand,
if $p$ lies outside $\hcube$, we insert it into or delete it from $\A$.
Once $\A$ has been updated $m$ times, we end the current epoch and discard the 
current $\kernel$. We begin a new epoch and reconstruct
$\A$, $P_I$, and $\kernel_I$ as described above. 

The preprocessing step at the  beginning of a new epoch causes $O(m)$ changes in $\kernel$
and there are at least $m$ updates in each epoch, therefore 
the algorithm maintains a 
stable kernel in the amortized sense. As above, using a de-amortization 
technique, we can ensure that $\kernel$ is stable. The correctness of the 
algorithm follows from the following lemma.

%\begin{figure}[t]
%%\center{\includegraphics{anchor-points} }
%\center{\includegraphics{figs/anchor-points.pdf}}
%\caption[Anchor points in $\b{R}^2$]{\label{fig:anchor-points} Anchor points $A_i$ and $A^\prime$ in $\b{R}^2$.} 
%\end{figure}

\begin{lemma}
\label{lem:anchor}
$\kernel$ is always an $\eps$-kernel of $P$.
\end{lemma}

\begin{proof}
It suffices to prove the lemma for a single epoch. Since we begin a new 
epoch after $m$ updates in $\A$, there is at least one $i$ such that
$A_i \subseteq \A \subseteq \kernel$. 
Thus we can show that $A_i \cup \kernel_I$ forms an $\eps$-kernel of 
$A_i \cup P_I$.  For any direction $u \in \b{S}^{d-1}$
\begin{eqnarray*}
\wid(\kernel_I,u) \leq \wid(P_I,u) &\leq& \wid(\kernel_I, u) + 2(1+\alpha)\delta \cdot \wid(A,u) 
\\ & \leq &
 \wid(\kernel_I, u) +  2(1+\alpha) \delta \beta_d \cdot \wid(P_I,u)  \textrm{\hspace{.3in} [via Lemma \ref{lemma:fat}]}
\\ & \leq &
 \wid(\kernel_I, u) +  2(1+\alpha) \delta \beta_d \cdot \wid(P_i,u)
\\ & \leq &
 \wid(\kernel_I, u) + 2(1+\alpha) \delta \beta_d \cdot \wid(I_i, u)
\\ & \leq &
 \wid(\kernel_I, u) + 2(1+\alpha)\delta (\beta_d)^2 \cdot \wid(A_i, u)   \textrm{\hspace{.3in} [via Lemma \ref{lemma:fat}]}
\\ & \leq &
 \wid(\kernel_I, u) + 2(1+\alpha)\delta (\beta_d)^2 \cdot \wid(P_I \cup A_i,u).
\end{eqnarray*}
Thus, for any direction $u \in \b{S}^{d-1}$ the extreme point of $P_I \cup A_i$ is either in $P_I$ or $A_i$.  In the first case, $\kernel_I$ approximates the width within a factor of $2(1+\alpha)\delta (\beta_d)^2 \cdot \wid(P_I \cup A_i, u) = \eps \cdot \wid(P_I \cup A_i, u)$.  In the second case, the extreme point is in $\kernel$ because all of $A_i$ is in $\kernel$.  
Thus the set $P_I \cup A_i$ has an $\eps$-kernel in $\kernel$, and the rest of the points are also in $\kernel$, so $\kernel$ is an $\eps$-kernel of the full set $P$.  

The size of $\kernel$ starts at $O(m)$ because both $\kernel_I$ and $\A$ are of size $O(m)$.  At most $m$ points are inserted outside of $I_{A}$ and hence into $\A$, thus the size of $\kernel = \kernel_I \cup \A$ is still $O(m)$ after $m$ steps.  Then the epoch ends.
\end{proof}

Using Lemmas~\ref{lem:weak-fixed} and~\ref{lem:strong-fixed}, we can
bound the update time and conclude the following.
\begin{lemma}
For a set $P$ of $n$ points in $\reals^d$ and a parameter 
$0 < \eps < 1$, there is a data structure that can maintain 
a stable $\eps$-kernel of $P$ of size:
\begin{itemize} \denselist
\item[(a)] $O(1/\eps^{(d-1)/2})$ under insertions and deletions in time 
$O(n \eps^{(d-1)/2} + 1/\eps^{(d-1)/2} + \log n)$, or
\item[(b)] $O(1/\eps^{d-1})$ in time $O(n \eps^{d-1} + \log n + \log (1/\eps))$.
\end{itemize}
\label{lem:stable-outter}
\end{lemma}

\begin{proof}
We build an outer kernel of size $O(m)$ in $O(n +  m \log n)$ time.  % using Lemma \ref{lem:m-anchor}.  
It lasts for $\Omega(m)$ insertions or deletions, so its construction time can be amortized over that many steps, and thus it costs $O(n/m + \log n)$ time per insertion or deletion.

In maintaining the inner kernel the preprocessing time can be amortized over $m$ steps, but the update time cannot.  
In case (a) we maintain the inner kernel of size $m = O(1/\eps^{(d-1)/2})$ with Lemma \ref{lem:strong-fixed}.  The update time is $O(n \eps^{(d-1)/2} + 1/\eps^{(d-1)/2})$.  
In case (b) we maintain the inner kernel of size $m = O(1/\eps^{d-1})$ with Lemma~\ref{lem:weak-fixed}.  The update time is $O(n \eps^{d-1} + \log (1/\eps))$.  

The update time can be made worst case using a standard de-amortization techniques~\cite{Ove83}.  More specifically, we start rebuilding the inner and outer kernels after $m/2$ steps and spread out the cost over the next $m/2$ steps.  We put all of the needed insertions in a queue, inserting a constant number of points to $\kernel$ each update to $P$.  Then after the new kernel is built, we enqueue required deletions from $\kernel$ and perform a constant number each update to $P$ over the next $m/2$ steps.  
\end{proof}

%%%%%%%%%%%%%%%%%%%%%%%%%%%%%%%%%%%%%%%%%%%%%
\paragraph{Putting it together.}

For a point set $P \subset \b{R}^d$ of size $n$, we can produce the 
best size and update time tradeoff for stable $\eps$-kernels by 
invoking Lemma \ref{lem:chain} to compose three stable $\eps$-kernel 
algorithms, as illustrated in Figure \ref{fig:compose}.  
We first apply Lemma \ref{lemma:chan} to maintain a stable 
$(\eps/3)$-kernel $\kernel_1$ of $P$ of size $O(\min\{n, (1/\eps^{d-1}) \log n\})$ 
with update time $O(\log n)$.  
We then apply Lemma~\ref{lem:stable-outter} to maintain a stable 
$(\eps/3)$-kernel $\kernel_2$ of $\kernel_1$ of size $O(1/\eps^{d-1})$ 
with update time $O(|\kernel_1| \eps^{d-1} + \log |\kernel_1| + \log (1/\eps)) 
= O(\log n + \log (1/\eps))$.  
Finally we apply Lemma~\ref{lem:stable-outter} again to maintain a 
stable $(\eps/3)$-kernel $\kernel$ of $\kernel_2$ of size 
$O(1/\eps^{(d-1)/2})$ with update time 
$O(|\kernel_2| \eps^{(d-1)/2} + 1/\eps^{(d-1)/2} + \log |\kernel_2|) 
= O(1/\eps^{(d-1)/2})$.  
$\kernel$ is a stable $\eps$-kernel of $P$ of size
$O(1/\eps^{(d-1)/2})$ with update time $O(\log n +
1/\eps^{(d-1)/2})$.  This completes the proof of
Theorem~\ref{theo:dynamic}.

%\begin{theorem}
%For a point set $P \subset \b{R}^d$ of size $n$ we can maintain a stable $\eps$-kernel with size $O(1/\eps^{(d-1)/2})$ under insertion and deletion in time $O(1/\eps^{(d-1)/2} + \log (\eps n))$.
%\end{theorem}

\begin{figure}[htb]
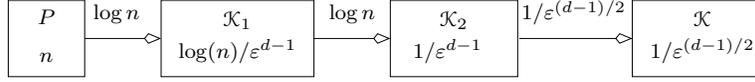

%\center{\includegraphics{pipeline.pdf}
\begin{center}
\input figs/pipeline
\caption[Composing stable $\eps$-kernel algorithms]{ Composing
stable $\eps$-kernel algorithms.} 
\label{fig:compose}
\end{center}
\end{figure}

\vspace{-.2in}
%%%%%%%%%%%%%%%%%%%%%%%%%%%%%%%%%%%%%%%%%%%%%%%%%%%
%%%%%%%%%%%%%%%%%%%%%%%%%%%%%%%%%%%%%%%%%%%%%%%%%%%
%%%%%%%%%%%%%%%%%%%%%%%%%%%%%%%%%%%%%%%%%%%%%%%%%%%
%%%%%%%%%%%%%%%%%%%%%%%%%%%%%%%%%%%%%%%%%%%%%%%%%%%
\section{Approximation Stability}
\label{sec:approx}

In this section we prove Theorem~\ref{theo:approx}. 
We first give a short proof for the lower-bound and then a more involved proof of the upper bound.  For the upper bound, we first develop basic ideas and prove the theorem in $\b{R}^2$ and $\b{R}^3$ before generalizing to $\b{R}^d$.

%%%%%%%%%%%%%%%%%%%%%%%%%%%%%%%%%%%%%%%%%%%%%%%%%%
\subsection{Lower Bound}
Take a cyclic polytope with $n$ vertices and $\Omega(n^{\lfloor d/2 \rfloor})$ facets and convert it into a fat polytope $\polyt$ using standard procedures~\cite{AHV04}.  For a parameter $\eps>0$, we add, for each facet $f$ of $\polyt$, a point $p_f$ that is $\eps$ far away from the facet. Let $P$ be the set of vertices of $\polyt$ together with the collection of added points. We choose $\eps$ sufficiently small so that points in $P$ are in convex position and all non-facet faces of $\polyt$ remain as faces of $\conv{P}$. Then the size of an optimal $\eps$-kernel of $P$ is at most $n$ (by taking the vertices of $\polyt$ as an $\eps$-kernel), but the size of an optimal $\eps/2$-kernel is at least the number of facets of $\polyt$, because every point of the form $p_f$ has to be present in the kernel. 
The first half of the lower bound is realized with $O(1/\eps^{(d-1)/2})$ evenly-spaced points on a sphere, and hence the full lower bound is proved.

%%%%%%%%%%%%%%%%%%%%%%%%%%%%%%%%%%%%%%%%%%%%%%%%%%
\subsection{Upper Bound}
By \cite{AHV04}, it suffices to consider the case in which $P$ is fat and the diameter of $P$ is normalized to $1$.  
Let $\kernel$ be an $\eps$-kernel of $P$ of the smallest size.  
Let $\polyt=\conv{\kernel}$, and  $\polyt_\eps=\polyt \oplus \eps \ball^d$.  We have $\polyt \subseteq \conv{P} \subseteq \polyt_\eps$ by the definition of $\eps$-kernels. It suffices to show that there is a set $\kernel' \subseteq P$ such that for $\polyt'=\conv{\kernel'}$, $\polyt' \subseteq \conv{P} \subseteq \polyt'_{\eps/2}$, and $|\kernel'| =O(|\kernel| ^{\lfloor d/2 \rfloor}\log^{d-2} (1/\eps))$ \cite{AHV04}.

For convenience, we assume that $\kernel'$ is not necessarily a subset of points in $P$; instead, we only require $\kernel'$ to be a subset of points in $\conv{P}$. By Caratheodory's theorem, for each point $x \in \kernel$, we can choose a set $P_x \subseteq P$ of at most $d+1$ points such that $x \in \conv{P_x}$.  We set $\bigcup_{x \in \kernel'} P_x$ as the desired $(\eps/2)$-kernel of $P$; $|\bigcup_{x \in \kernel'} P_x| \leq (d+1)|\kernel^\prime| = O(\opt{P}{\eps}^{\lfloor d/2 \rfloor} \log^{d-2} (1/\eps))$.

Initially, we add a point into $\kernel'$ for each point in $\kernel$.  If $p \in \kernel$ lies on $\partial \conv{P}$, we add $p$ to $\kernel'$.  Otherwise we project $p$ onto $\partial \conv{P}$ in a direction in which $p$ is maximal in $\kernel$ and add the projected point to $\kernel'$.  Abusing the notation slightly, we use $\polyt$ to denote the convex hull of these initial points.  
%We also assume that each $p \in \kernel$ is on the boundary of $\conv{P}$, otherwise, for each point $p \in \kernel$, we can pick a direction $u \in \b{S}^{d-1}$ in which it is maximal in $\kernel$ and project $p$ onto $\conv{P}$ in direction $u$.  Let these projected points be the initial points in $\kernel'$ instead of the points in $\kernel$.  
For simplicity, we assume $\polyt$ to be a simplicial polytope. 
%The proof can be extended to a non-simplicial polytope by triangulating each non-simplicial face into simplicies and then applying a similar argument.

%%% structure
\paragraph{Decomposition of $\polyt_\eps \setminus \intr \polyt$.}
There are $d$ types of simplices on $\partial \polyt$.  In $\b{R}^2$ these are points and edges.  In $\b{R}^3$ these are points, edges, and triangles.  We can decompose $\polyt_\eps \setminus \intr \polyt$ into a set of regions, each region $\sigma(f)$ corresponding to a simplex $f$ in $\polyt$.  %; see an example in Figure \ref{fig:2d}(a).  
For each simplex $f$ in $\polyt$ let $\dual{f} \subseteq \b{S}^{d-1}$ denote the dual of $f$ in the Gaussian diagram of $\polyt$.  
Recall that if $f$ has dimension $k$ ($0 \leq k \leq d-1$), then $\dual{f}$ has dimension $d-1-k$.  
The region $\polyt_\eps \setminus \intr \polyt$ is partitioned into a collection of $|\polyt|$ regions (where $|\polyt|$ is the number of faces of all dimensions in $\polyt$).  Each simplex $f$ in $\polyt$ corresponds to a region defined
$$
\sigma(f) = \{ f + z u \mid 0 \leq z \leq \eps,\, u \in \dual{f}\}.
$$
For a subsimplex $\tau \in f$, we can similarly define a region 
$
\sigma(\tau) = \{\tau + zu \mid 0 \leq z \leq \eps,\, u\in \dual{f}\}.
$
In $\b{R}^2$, there are two types of regions: point regions and edge regions.  
In $\b{R}^3$, there are three types of regions: point regions (see Figure \ref{fig:decomposition}(a)), edge regions (see Figure \ref{fig:decomposition}(b)), and triangle regions (see Figure \ref{fig:decomposition}(c)).  

For convenience, for any point $q =\fq + z\cdot u \in \sigma(f)$, where $\fq\in f, 0\leq z\leq \eps$, and $u\in \dual{f}$, we write $q = \point{\fq}{z}{u}$ (which intuitively reads, the point whose projection onto $f$ is $\fq$ and which is at a distance $z$ above $f$ in direction $u$). 
We also write $\rotate{q}{v} = \fq+z\cdot v$ (intuitively, $\rotate{q}{v}$ is obtained by rotating $q$ w.r.t.~$f$ from direction $u$ to direction $v$).
Similarly, we write a simplex $\point{\fDelta}{z}{u}=\fDelta \oplus z\cdot u$, where $\fDelta$ is a simplex inside $f$, $0\leq z\leq \eps$, and $u\in \dual{f}$, and write $\rotate{\Delta}{v} = \fDelta \oplus z\cdot v$.

We will proceed to prove the upper bound as follows.  For each type of region $\sigma(f)$ we place a bounded number of points from $\sigma(f) \cap \conv{P}$ into $\kernel'$ and then prove that all points in $\sigma(f) \cap \conv{P}$ are within a distance $\eps/2$ from some point in $\polyt' = \conv{\kernel'}$.  
We begin by introducing three ways of ``gridding'' $\sigma(f)$ and then use these techniques to directly prove results for several base cases, which illustrate the main conceptual ideas.  These base cases will already be enough to prove the results in $\b{R}^2$ and $\b{R}^3$.  Finally we generalize this to $\b{R}^d$ using an involved recursive construction.  
We set a few global values: $\delta = \eps/12d$, $\theta = 2\arcsin (\delta/2\eps)$, and $\rho = \delta/\eps$.

\begin{description} \denselist
\item[1: Creating layers.]
For a point $q = \point{\fq}{z}{u} \in \sigma(f)$ we classify it depending on the value $z = |q-\fq|$.  If $z \leq \eps/2$, then $q$ is already within $\eps/2$ of $\polyt$.  We then divide the range $[\eps/2, \eps]$ into a constant $H = (\eps/2)/\delta$ number of cases using $\c{H} = \{ h_1 = \eps/2, h_2 = \eps/2+\delta, \ldots, h_H = \eps-\delta \}$.  If $z \in [h_i, h_{i+1})$, then we set $q_{h_i} = \point{\fq}{h_i}{u}$.  We define $\Psi_{f,h_i} \subset \sigma(f) \cap \conv{P}$ to be the set of points that are a distance exactly $h_i$ from $f$.  
\vspace{1mm}

\item[2: Discretize angles.]
We create a constant size $\theta$-net $U_{f,h} = \{u_1, u_2, \ldots\} \subset \dual{f}$ of directions with the following properties.
(1) For each $q = \point{\fq}{h}{u} \in \Psi_{f,h}$ there is a direction $u_i \in U_{f,h}$ such that the angle between $u$ and $u_i$ is at most $\theta$.
(2) For each $u_i \in U_{f,h}$ there is a point $p_i = \point{\fp_i}{h}{u_i} \in \Psi_{f,h}$; let $N_{f,h} = \{p_i \mid i \geq  1\}$.  
%Let $N_{f,h} = \{p_1, p_2, \ldots\} \subset \Psi_{f,h}$ so that $p_i = \point{\bar{p}_i}{h}{u_i}$ (for $u_i \in U_{f,h}$) and for any point $q_h = \point{\fq}{h}{u} \in \Psi_{f,h}$ there is a point $p_i = \point{\bar{p}_i}{h}{u_i} \in N_{f,h}$ such that the angle between $u$ and $u_i$ is at most $\theta$.  
$U_{f,h}$ is constructed by first taking a $(\theta/2)$-net $U_f$ of $\dual{f}$, then for each $u_i' \in U_f$ choosing a point $p_i = \point{\fq_i}{h}{u_i} \in \Psi_{f,h}$ where $u_i$ is within an angle $\theta/2$ of $u'_i$ (if one exists), and finally placing $u_i$ in $U_{f,h}$.  
\vspace{1mm}

\item[3: Exponential grid.]
Define a set $\c{D} = \{d_0, d_1 = (1+\rho) d_0, \ldots, d_m = (1+\rho)^m d_0\}$ of distances where $d_m < 1$ and $d_0 = \delta$, so $m = O(\log 1/\eps)$.  
%Choose $r = \point{\fr}{h}{u} \in \point{f}{h}{u} \cap \conv{P}$ and call it the \emph{support point of $f$, given $h$ and $u$}.  
For a face $f \in \polyt$, let any $r \in \sigma(f)$ be called a \emph{support point of $f$}.  
%For each boundary facet $F$ of $f$ we construct a set of at most $m$ homothets of $F$ as follows.  
Let $p_1, \ldots, p_k$ be the vertices of the $k$-simplex $f$.  For each $p_j$, and each $d_i \in \c{D}$ (where $d_i < ||p_j - \fr||$), let $p_{j,i}$ be the point at distance $d_i$ from $p_j$ on the segment $p_j \fr$.  
For each boundary facet $F$ of $f$, define a sequence of at most $m$ simplices $F_0, F_1, \ldots \in \conv{F \cup \fr}$, each a homothet of $F$, so the vertices of $F_i$ lie on segments $p_j \fr$ where $p_j \in \partial F$ (see Figure \ref{fig:correct-triangle}(a)).  
The translation of each $F_i$ is defined so it intersects a point $p_{j,i}$ (where $p_j \in \partial F$) and is as close to $F$ as possible.  This set of $(k-1)$-simplices for each $F$ defines the exponential grid $G_{r,f}$.
%Let $F_i$ be the translated copy of $F$ that passes through $p_{j,i}$, for the $j$ such that $p_{j,i}$ is closest to $F$.  The family of translated (and truncated) facets $\{F_i \cap f\}_i$ for each boundary facet $F$ of $f$ forms the exponential grid $G_{r, f}$.  
The full grid structure is revealed as this is applied recursively on each $F_i$.

The exponential grid $G_{r,\Delta}$ on a simplex $\Delta$ has two important properties for a point $q \in \Delta$:  
\begin{itemize} \denselist
\item[(G1)] If $q \in \conv{F \cap \fr}$ lies between boundary facet $F$ and $F_0$, let $q_0$ be the intersection of the line segment $q \fr$ with $F_0$; then $||q - q_0|| \leq d_0 = \delta$.
\item[(G2)] If $q \in \conv{F \cap \fr}$ lies between $F_{i-1}$ and $F_i$ and the segment $q \fr$ intersects $F_i$ at $q_i$, let $q_F$ be the intersection of $F$ with the ray $\overrightarrow{r q}$; then $||q_i - q|| / ||q_i - q_F|| \leq \rho = \delta/\eps$.  
\end{itemize}
\end{description}

We now describe how to handle certain simple types of regions: where $f$ is a point or an edge.  These will be handled the same regardless of the dimension of the problem, and they (the edge case in particular) will be used as important base cases for higher dimensional problems.  

%%% point regions.  %%%%%%%%%%%%%%%%%%%%%%%%%%%%%%
\paragraph{Point regions.}
Consider a point region $\sigma(p)$.  
For each $h \in \c{H}$ create $\theta$-net $U_{p,h}$ for $\Psi_{p,h}$, so $N_{p,h}$ are the corresponding points where each $p_i = \point{p}{u_i}{h} \in N_{p,h}$ has $u_i \in U_{p,h}$.   Put each $N_{p,h}$ in $\kernel'$.  

\parpic[r]{\PPic{1.5cm}{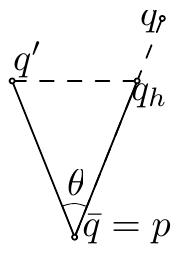}} 
For any point $q = \point{\fq}{z}{u'} \in \sigma(p) \cap \conv{P}$, let $q' = \point{\fq}{h}{u}$ where $h \in \c{H}$ is the largest value such that $h \leq z$ and $u \in U_{p,h}$ is the closest direction to $u'$; set $q_h = \point{\fq}{h}{u'} = \rotate{q'}{u'}$.  
First $||q - q_h|| \leq \delta$ because $z-h \leq \delta$.  
Second $||q_h - q'|| \leq \delta$ because the angle between $u'$ and $u$ is at most $\theta$, and they are rotated about the point $p$.  
%That is $||q_h - q'|| \leq 2 \eps \sin (\theta/2) \leq 2 \eps \sin \arcsin (\delta/2\eps) = \delta$.  
Thus $||q - q'|| \leq ||q - q_h|| + ||q_h - q'|| \leq 2\delta \leq \eps/2$.  

\begin{lemma}
\label{lem:point-reg}
%Let $\kernel \subset P \subset \b{R}^d$ be an $\eps$-kernel of $P$.
For a point region $\sigma(p)$, there exists a constant number of points $K_p \subset \sigma(p) \cap \conv{P}$ such that all points $q \in \sigma(p) \cap \conv{P}$ are within a distance $\eps/2$ of $\conv{K_p}$.  
\end{lemma}

%%% edge regions.  %%%%%%%%%%%%%%%%%%%%%%%%%%%%%%
\paragraph{Edge regions.}
Consider an edge region $\sigma(e)$ for an edge $e$ of $\polyt$.  
Orient $e$ along the $x$-axis.  
For each $h \in \c{H}$ and $u \in U_{e,h}$, let $\Psi_{e,h,u}$ be the set of points in $\Psi_{e,h}$ within an angle $\theta$ of $u$.  For each $\Psi_{e,h,u}$, we add to $K_e$ the (two) points of $\Psi_{e,h,u}$
%place the points $p^+_{h,u}$ and $p^-_{h,u}$ 
with the largest and smallest $x$-coordinates, denoted by $p^+_{h,u}$ and $p^-_{h,u}$.

\begin{figure}[htb!]
\centering{ \small 
\begin{tabular}{ccc}
\includegraphics[scale=1]{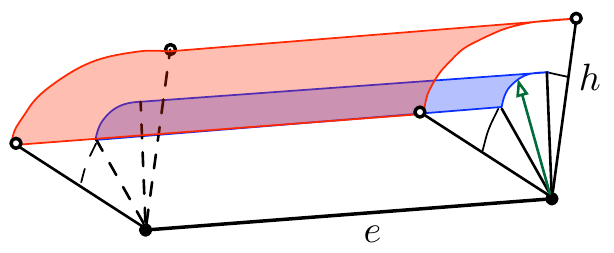} & \hspace{0.5cm} & 
\includegraphics[scale=1]{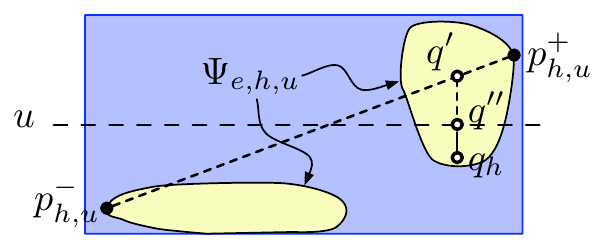}\\
(a) $\sigma(e)$ in $\b{R}^3$ && (b) top view of $\sigma(e)$ at height $h$
\end{tabular}
}
\caption{Illustration of 2 points in $\kernel'$ for edge case with specific $h \in \c{H}$ and $u \in U_{e,h,\theta}$. }
\label{fig:edge-case}
\end{figure}

For any point $q = \point{\fq}{z}{v} \in \sigma(e) \cap \conv{P}$, there is a point $q'' = \point{\fq}{h}{u}$ such that $h \in \c{H}$ is the largest value less than $z$ and $u \in U_{e,h}$ is the closest direction to $v$.  Furthermore, $||q - q''|| \leq ||q - q_h|| + ||q_h - q''|| \leq (z-h) + 2 \eps \sin(\theta/2) = \delta + \delta = 2\delta$.  
We can also argue that there is a point $q' = \point{\fq}{z'}{u'} \in p^-_{h,u} p^+_{h,u}$, because if $\fq$ has smaller $x$-coordinate than $\fp^-_{h,u}$ or larger $x$-coordinate than $\fp^+_{h,u}$, then $q'$ cannot be in $\Psi_{e,h,u}$.  
Clearly the angle between $u$ and $u'$ is less than $\theta$.  This also implies that $h-z' < \delta$.  Thus $||q'' - q'|| \leq 2\delta$, implying $||q-q'|| \leq 4 \delta \leq \eps/2$.

\begin{lemma}
\label{lem:edge-reg} 
%Let $\kernel \subset P \subset \b{R}^d$ be an $\eps$-kernel of $P$.
For an edge region $\sigma(e)$, there exists $O(1)$ points $K_e \subset \sigma(e) \cap \conv{P}$ such that for any point $q = \point{\fq}{v}{z} \in \sigma(e) \cap \conv{P}$ there is a point $p= \point{\fq}{u}{h} \in \conv{K_e}$ such that $z-h \leq 2\delta$, $||v-u||\leq 2\delta$, and, in particular, $||q-p|| \leq 4\delta \leq \eps/2$.  
\end{lemma}

%%%%%%%%%%%%%%%%%%%%%%%%%%%%%%%%%%%%%%%%%%%%%%%%%%
\subsubsection{Approximation Stability in $\b{R}^2$}

For $\kernel \subset P \in \b{R}^2$ there are $|\kernel|$ points and edges in $\polyt$.  
Thus combining Lemmas~\ref{lem:point-reg} and~\ref{lem:edge-reg} $|\kernel'|/|\kernel| = O(1)$ and we have proven Theorem~\ref{theo:approx} for $d=2$.  

\begin{theorem}
For any point set $P \in \b{R}^2$ and for any $\eps > 0$ we have 
$
\kappa(P, \eps/2) / \kappa(P,\eps) = O(1).
$
\end{theorem}

%%%%%%%%%%%%%%%%%%%%%%%%%%%%%%%%%%%%%%%%%%%%%%%%%%
\subsubsection{Approximation Stability in $\b{R}^3$}

%%%%%%%%%%%%%%%%%%%%%%%%%
\paragraph{Construction of $\kernel'$.}
Now consider $\kernel \subset P \in \b{R}^3$ and the point regions, edge regions, and triangle regions in the decomposition of $\polyt_\eps \setminus \intr \polyt$ (see Figure \ref{fig:decomposition}).  By Lemmas \ref{lem:point-reg}  and \ref{lem:edge-reg} we can add $O(|\kernel|)$ points to $\kernel'$ to account for all point and edge regions.  We can now focus on the $O(|\kernel|)$ triangle regions.

\begin{figure}[htb!]
\begin{center}
{\small 
\begin{tabular}{ccccc}
\includegraphics[scale=1]{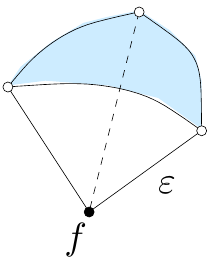} 
& \hspace{1cm} & 
\includegraphics[scale=1]{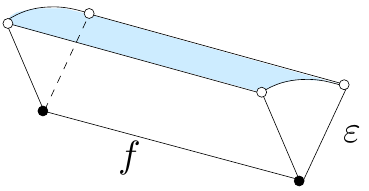} 
& \hspace{1cm} & 
\includegraphics[scale=1]{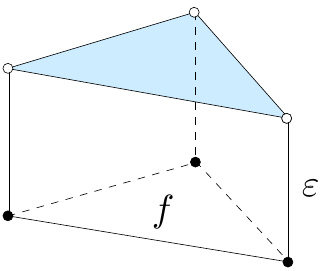}
\\
(a) $f$ is a vertex of $\polyt$ && (b) $f$ is an edge of $\polyt$ && (c) $f$ is a facet of $\polyt$
\end{tabular}
}
\end{center}
%\vspace{-.2in}
\caption{Illustration of regions in the partition of $\polyt_\eps \setminus \intr \polyt$ in three dimensions.}
\label{fig:decomposition}
\end{figure}

%%% triangle regions
Consider a triangle region $\sigma(t)$ for a triangle $t$ in $\polyt$ (see Figure \ref{fig:correct-triangle}(a)), $\dual{t}$ consists of a single direction, the one normal to $t$. 
%For all layers $h \in \c{H}$, 
Let $r$ be the highest point of $\sigma(t) \cap \conv{P}$ in direction $\dual{t}$.  We add $r$ to $\kernel'$ and we create an exponential grid $G_{r,t}$ with $r$ as the support point.
%a support point of $t$, and put it in $\kernel'$.  We then create an exponential grid $G_{\fr,t}$. 
For each edge $e \in G_{r,t}$ and $h \in \c{H}$ we add the intersection of $\point{e}{h}{t^*}$ with the boundary of $\sigma(t) \cap \conv{P}$ to $\kernel'$, as shown in Figure \ref{fig:correct-triangle}(b).  
Thus, in total we add $O(|\kernel| \log (1/\eps))$ points to $\kernel'$.

%%%%%%%%%%%%%%%%%%%%%%%%%
\paragraph{Proof of correctness.}

\begin{figure}[htb!]
\begin{center}
{\small 
\begin{tabular}{ccccc}
\includegraphics[scale=1]{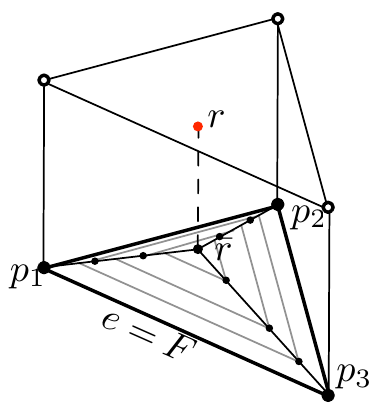} 
& \hspace{0.1cm} & 
\includegraphics[scale=1]{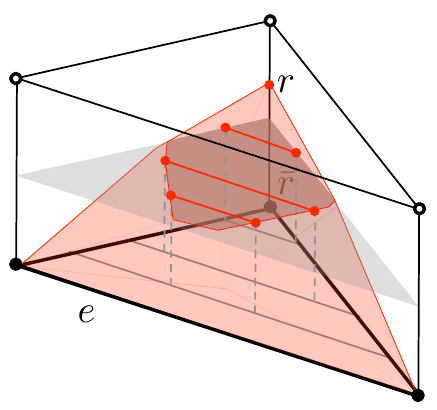} 
& \hspace{0.1cm} & 
\includegraphics[scale=1]{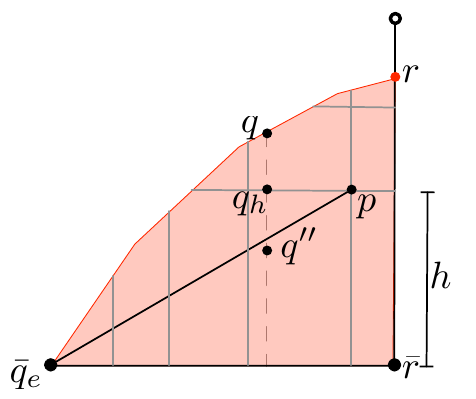}
\\
(a) $\sigma(t)$ with $r$ and $G_{r,t}$ && (b) Subtriangle $t_e$ of $\sigma(t)$  && (c) Slice of (b) through $r$, $q$
\end{tabular}
}
\end{center}
%\vspace{-.2in}
\caption{Illustration to aid correctness of approximation of triangle regions in $\b{R}^3$.}
\label{fig:correct-triangle}
\end{figure}

%%% triangle regions
Consider any point $q = \point{\fq}{z}{t^*} \in \sigma(t) \cap \conv{P}$ and associate it with a boundary edge $e$ of $t$ such that $\fq \in \conv{e \cup \fr}$.  Let $q_h = \point{\fq}{h}{t^*}$ where $h \in \c{H}$ is the largest height such that $h \leq z$.  
If segment $\fq \fr$ does not intersect any edge $e_i$ parallel to $e$ in $G_{r,t}$, let $\fp = \fr$.  
Otherwise, let $e_i$ be the first segment parallel to $e$ in $G_{r,t}$ intersected by the ray $\overrightarrow{\fq \fr}$, and let $\fp$ be the intersection.  
Let $p = \point{\fp}{h}{t^*}$ which must be in $\conv{\kernel'}$ by construction.  
If $e_i = e_0$, then by (G1) we have $||q_h - p|| = ||\fq - \fp|| \leq \delta$, thus $||q - p|| \leq 2\delta \leq \eps/2$ and we are done.  
Otherwise, let $\fq_e$ be the intersection of $e$ with ray $\overrightarrow{\fr \fq}$.  By (G2) $||\fp - \fq||/||\fp - \fq_e|| \leq \rho = \delta/\eps$.  Thus, $q'' = \point{\fq}{h-\eps\rho}{t^*}$ is below the segment $\fq_e p$ (see Figure \ref{fig:correct-triangle}(c)) and thus $q'' \in \conv{\kernel'}$ since triangle $p \fp \fq_e$ is in $\conv{\kernel'}$.  Finally, $||q - q''|| = ||q - q_h|| + ||q_h - q''|| \leq 2\delta \leq \eps/2$.  
This proves Theorem \ref{theo:approx} for $d=3$.  

\begin{theorem}
For any point set $P \in \b{R}^3$ and for any $\eps > 0$ we have 
$
\kappa(P, \eps/2) / \kappa(P,\eps) = O(\log 1/\eps).
$
\end{theorem}

%%%%%%%%%%%%%%%%%%%%%%%%%%%%%%%%%%%%%%%%%%%%%%%%%%%
\subsubsection{Approximation Stability in $\b{R}^d$}

%%%%%%%%%%%%%%%%%%%%%%%%%
\paragraph{Construction of $\kernel'$.}
The number of regions in the decomposition of $\polyt_\eps \setminus \intr \polyt$ is $|\polyt| = O(|\kernel|^{\lfloor d/2 \rfloor})$~\cite{Zie95}.  For each region $\sigma(f)$, we choose a set $K_f$ of $O(\log^{d-2} (1/\eps))$ points such that 
$
\conv{K_f} \subseteq \conv{P} \cap \sigma(f) \subseteq \conv{K_f} \oplus (\eps/2)\ball^d.
$ 
Then we set $\kernel' = \bigcup_f K_f$.

When $f$ is a $0$- or $1$-simplex, we apply Lemmas \ref{lem:point-reg}  and \ref{lem:edge-reg}.  
Otherwise, to construct $K_f$ we create a recursive exponential grid on $f$.  
Specifically, for all $h \in \c{H}$ and $u \in U_{f,h}$ we choose a point $r = \point{\fr}{h}{u} \in \sigma(f) \cap \conv{P}$ and construct an exponential grid $G_{r,f}$ with $r$ as the support point.  
Next, for each $\Delta \in G_{r,f}$ we recursively construct exponential grids.  That is, for all $h' \in \c{H}$ and $u' \in U_{\Delta,h'}$ we choose another support point $r' = \point{\fr'}{h'}{u'}$ and construct an exponential grid $G_{r',\Delta}$ on $\Delta$.  
%Then for each $\Delta$ in each $G_{r_{f,h,u},f}$, and for each $h' \in \c{H}$ and $u' \in U_{\Delta,h'}$ we choose another support point $r_{\Delta,h',u'}$ and construct another exponential grid $G_{r_{\Delta,h',u'}, \Delta}$.  
At each iteration the dimension of the simplex in the exponential grids drops by one.  We continue the recursion until we get $1$-simplices.  
Let $\c{G}_f$ be the union of all exponential grids.  

Let $e$ be a $1$-simplex in $\c{G}_f$.  For each height $h \in \c{H}$ and direction $u \in U_{e,h}$ we choose two points as described in the construction the edge region and add them to $K_f$.  
%At each each $1$-simplex $\Delta_1 \in \c{G}_f$ (a segment) we invoke Lemma \ref{lem:edge-reg} to generate $O(1)$ points and place them in $K_f$.  
We also place the support point of each $\Delta \in \c{G}_f$ into $K_f$.  
By construction, for a $k$-simplex $f$, $\c{G}_f$ contains $O(\log^{k-1} (1/\eps))$ simplices and thus $|K_f| = O(\log^{k-1} (1/\eps))$.  Hence $|\kernel'| = O(|\kernel|^{\lfloor d/2 \rfloor} \log^{d-2} (1/\eps))$.  
%By construction of the exponential grid, we have at most $O(\log^{k-1} (1/\eps))$ simplices in $\c{G}_f$ for a $k$-simplex $f$.  Since $k \leq d-1$, and each simplex results in $O(1)$ points in $K_f$, we have $|K_f| = O(\log^{d-2} (1/\eps))$.  

%%%%%%%%%%%%%%%%%%%%%%%%%
\paragraph{Proof of correctness.} 
Let $f$ be a $k$-face of $\polyt$ ($0\leq k\leq d-1$). We need to show for any point $q = \point{\fq}{z}{u} \in \conv{P} \cap \sigma(f)$, there is a point $p \in \conv{K_f}$ (specifically $p''_k$) such that $||q - p|| \leq \eps/2$.  

Before describing the technical details (mainly left to the appendix), we first provide some intuition regarding the proof.  For any $q \in \conv{P} \cap \sigma(f)$ we first consider $q_k = \point{\fq}{h}{u}$ where $h \in \c{H}$ is the largest $h<z$.  If $q_k \in \conv{K_f}$ we are done.  If $q_k \notin \conv{K_f}$, we need to find a ``helper point" $q_{k-1}$ for $q_k$.  If $q_{k-1} \notin \conv{K_f}$ we need to recursively find a ``helper point'' $q_{k-2}$ for $q_{k-1}$, and so on until $q_m \in \conv{K_f}$.  
Tracing back along the recursion, we can then prove that $q_k$ (and hence $q$) has a nearby point in $\conv{K_f}$.
Note that we do not prove $q_j$ is near $q_{j-1}$.  Formally:

\begin{lemma}
\label{lem:helper}
We can construct a sequence of helper points $q_j = \point{\fq_j}{h_j}{u_j}$ and simplices $\Delta_j \in \c{G}_f$ with the following invariants:
%Formally, for $j=k,k-1,\cdots$, we now make a sequence of recursive definitions, for a point $q_j = \point{\fq_j}{h_j}{u_j}$ and a simplex $\Delta_j \in \c{G}_f$ such that $\fq_j \in \Delta_j$.  
%We will maintain the following invariants for each index $j$: 
We can construct a sequence of helper points $q_j = \point{\fq_j}{h_j}{u_j}$ and simplices $\Delta_j \in \c{G}_f$ with the following invariants:
   \begin{itemize} \denselist 
     \item[\emph{(I1)}] $q_j\in \conv{P}$;
     \item[\emph{(I2)}] $h_j \in \c{H}$ and $h_j - h_{j+1} \leq \delta$ (for $j\neq k$);
     \item[\emph{(I3)}] $\fq_j \in \Delta_j \in \c{G}_f$ and the dimension of $\Delta_j$ is $j$; and 
     \item[\emph{(I4)}] $\| u_j -u_{j+1}\| \leq 2\theta$ (for $j\neq k$). 
   \end{itemize}
\end{lemma}

\begin{proof} 
Set $q_k$ as above and $\Delta_k = f$.    
Assume that for an index $j\leq k$, $q_j$ has been defined.  
Since (I1) implies $q_j \in \Psi_{\Delta_j, h} \neq \emptyset$, 
there is a direction $u_{j-1}' \in U_{f,h_j}$ such that the angle between $u_j$ and $u_{j-1}'$ is at most $\theta$.  
Let $r_j \in K_f$ be the support point of $\Delta_j$, given $h_j$ and $u_{j-1}'$. 
Let $F$ be the facet of $\Delta_j$ such that $\fq_j \in \conv{F \cup \fr_j}$.  If the segment $\fq_j \fr_j$ does not intersect any simplex in the family of $(j-1)$ simplices induced by $F$ and $r_j$, then let $q_{j-1} = r_j$ and terminate the recursion (see Figure \ref{fig:recursive_def}(a)).  
Otherwise, let $\Delta_{j-1}$ be the first such simplex intersected by the ray $\overrightarrow{\fq_j \fr_j}$, and let $\fq_{j-1}$ be the intersection point (see Figure \ref{fig:recursive_def}(b)).  
  
\begin{figure}[htb!]
\begin{center}
{\small 
\begin{tabular}{ccc}
\includegraphics[scale=1]{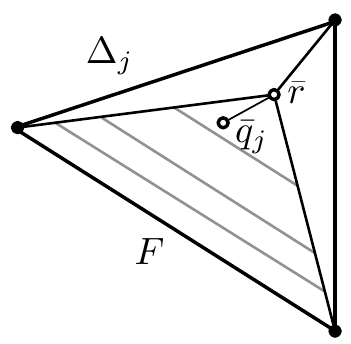} 
& \hspace{1cm} & 
\includegraphics[scale=1]{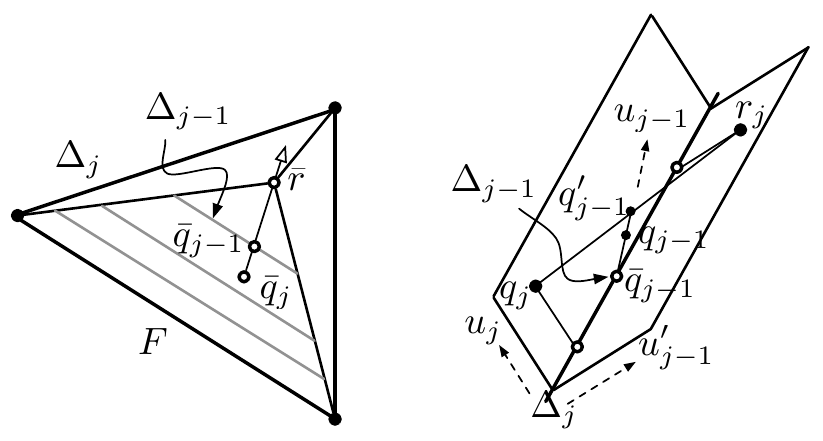} \\
(a)  && (b) 
\end{tabular}
}
\end{center}
\caption{(a) Terminate early: $q_{j-1} = r_j$. (b) Recursive case: $q_{j-1}'$ lies on the segment $q_j r_j$. }
\label{fig:recursive_def}
\end{figure}

To determine $q_{j-1}$ we first find the direction $u_{j-1}$, such that $q_{j-1}' = \point{\fq_{j-1}}{z'}{u_{j-1}}$ lies on the segment $q_j r_j$.  Then $h_{j-1}$ is determined as the maximum $h \in \c{H}$ such that $h \leq z'$. 
We can show (I1) is satisfied because $q_{j-1}'$ must be in $\conv{P}$ because $q_j$ and $r_j$ are, and then $q_{j-1} = \point{\fq_{j-1}}{h_{j-1}}{u_{j-1}}$ lies on segment $q_{j-1}' \fq_{j-1}'$.  
We show (I4) by $||u_j - u_{j-1}|| \leq ||u_j - u_{j-1}'|| \leq \theta$.  Invariant (I2) follows because $z' \geq h_j \cos \theta$ and thus $h_j - z' \leq h_j (1-\cos \theta) \leq \eps (\sin \theta) \leq \delta$.  
Invariant (I3) holds by construction.   

Assume the recursion does not terminate early.  At $j=1$, since $q_1 \in \conv{P}$ and $\Delta_1$ is a line segment, we can apply Lemma \ref{lem:edge-reg} to $\sigma(\Delta_1)$ and find a point $q_0 = \point{\fq_1}{h_1}{u_0} \in \conv{K_f}$ such that $||q_1 - q_0|| \leq 2\delta$ and $||u_1 - u_0|| \leq 2\theta$.  
% 
%Assume the recursion does not terminate early. For the case $k< d-1$,  we stop the recursion at $j=0$ and set $p_0=r_0 \in K_f$.   
%For the case $k=d-1$, we stop the recursion at $j=1$. Note that for $k=d-1$, since $\dual{f}$ is a singleton set and we have included the end points $q_1$ and $q_2$ of the line segment $\Delta_1 \cap \conv{P}$ in $K_f$, then $\Delta_1 \cap \conv{P} \in \conv{K_f}$.  Since $p_1 \in \conv{P}$ and $p_1 \in \Delta_1$, then $p_1 \in \conv{K_f}$. 
%
This completes the description of the recursive definition of helper points.  
\end{proof}

Let $q_m$ be the last point defined in the sequence. By construction, $q_m \in \conv{K_f}$.  
For each $m \leq j \leq k$, let $p_j = \rotate{q_j}{u_m}$.
We have the following key lemma, which shows that $p_j$ is close to a point $p''_j \in \conv{K_f}$.

\begin{lemma} \label{lem:key}
 For each $j\geq m$, there is a point $p_j'\in \point{\Delta_j}{h_j}{u_m}$ such that 
 \begin{itemize} \denselist 
   \item[(1)] $|| p_j' - p_j|| \leq j \delta$; and  
   \item[(2)] $p_j'' = p_j' - 2(j-m) \delta u_m \in \conv{K_f}$. 
 \end{itemize}
\end{lemma}

\begin{proof}
We prove the lemma by induction on $j$. For $j=m$,  since $p_m=q_m \in \conv{K_f}$, the claim is trivially true by setting $p_m'=p_m$.  Assume the claim is true for some $j\geq m$.  Now consider the case $j+1$.  Let $\fp_{F,j}$ be the intersection of the ray $\overrightarrow{\fq_j\fq_{j+1}}$ with $\partial f$ on facet $F$. 
%Let $x_j$ be the projection of $p_j'$ onto $f$, and 
Let $x_{j+1}$ be the intersection of $\fp_{F,j} \fp_j'$ with the line passing through $\fq_{j+1}$ and parallel to $\overline{\fq_j \fp_j'}$ (see Figure~\ref{fig:induction}). There are two cases:
 
 \begin{figure}[htb!]
\begin{center}
\includegraphics[scale=1]{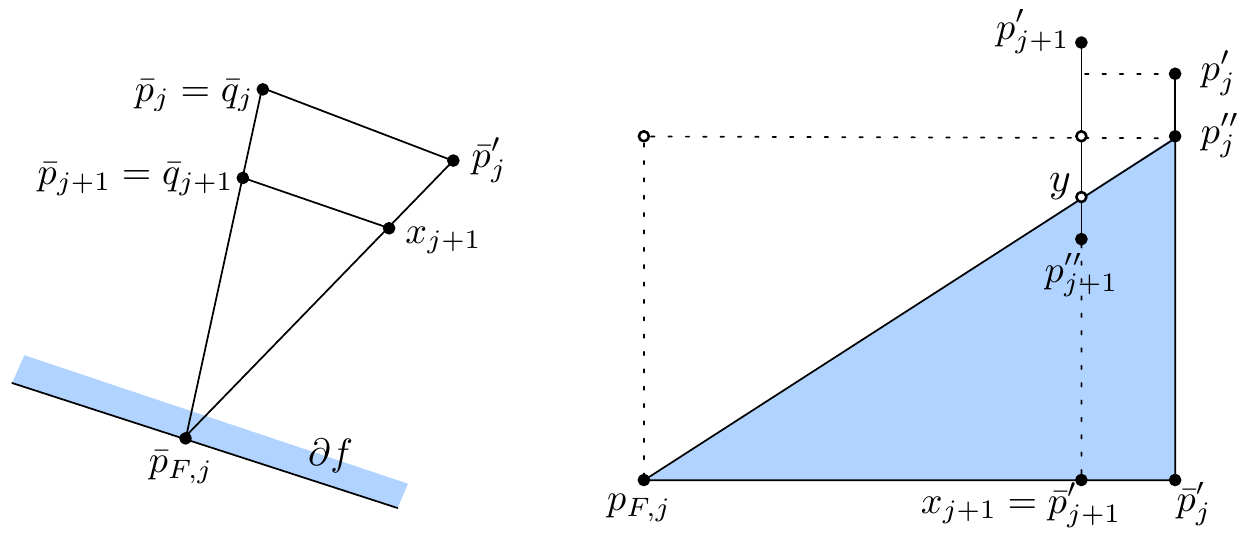}
\end{center}
\caption{The inductive step for proving $||p_{j+1}' - p_{j+1}|| \leq (j+1) \delta$ and $p_{j+1}'' \in \conv{K_f}$.}
\label{fig:induction}
\end{figure}
 
 \smallskip 
 {\em Case 1:} If $\Delta_j$  (such that $\fq_j \in \Delta_j$ according to (I3)) is the closest facet to $F$, then $q_{j+1}$ lies between $F$ and $\Delta_j$.  Thus by (G1), we know that  $||\fq_{j+1} - \fq_j || \leq \delta$. We set $p_{j+1}' = \point{\fp_j'}{h_{j+1}}{u_m}$. As such, 
$$
|| p_{j+1}' - p_{j+1} || = ||\fp_j' -  \fq_{j+1}|| \leq || \fp_j' - \fq_j|| + ||\fq_j - \fq_{j+1}|| \leq j\delta + \delta= (j+1)\delta.
$$
Moreover, since $h_{j+1} -  2(j+1-m)\delta \leq h_j- 2(j-m)\delta$, $p_{j+1}''$ lies on the segment $p_j'' \fp_j'$ and therefore $p_{j+1}''\in\conv{K_f}$.

 \smallskip  
 {\em Case 2:} Otherwise. In this case, we have by (G2)
$$
\frac{||\fp_j' - x_{j+1}||}{||\fp_j' - \fp_{F,j}||} = \frac{|| \fq_j - \fq_{j+1}||}{|| \fq_j - \fp_{F,j}||} \leq  \rho.
$$ 
We set $p_{j+1}' = \point{x_{j+1}}{h_{j+1}}{u_m}$. First observe that
$$
|| p_{j+1}' - p_{j+1} || = || x_{j+1} -  \fq_{j+1} || \leq ||\fp_j' - \fq_j|| \leq \delta. 
$$
Furthermore, let $y$ be the intersection of $\fp_{F,j} p_j''$ with $p_{j+1}'x_{j+1}$ (see Figure~\ref{fig:induction}(right)). We then have
$$
|| p_{j+1}' - y || \leq  (h_{j+1}-h_j) + ||p_j' - p_j''||  + \rho \cdot ||p_j'' - \fp_j'|| < \delta + 2(j-m)\delta +\delta = 2(j+1-m)\delta.
$$
  Therefore $p_{j+1}''$ lies below $y$ and as such $p_{j+1}'' \in \conv{K_f}$ since triangle $\fp_j' \fp_{F,j} p_j'' \in \conv{K_f}$.  
\end{proof}

We now complete the proof of Theorem \ref{theo:approx} with the following lemma.

\begin{lemma}
$|| q - p_k''|| \leq \eps/2$.  %\hspace{.2in} \emph{(Completing the proof of Theorem \ref{theo:approx}.)}
\end{lemma}

\begin{proof}
%$
%|| q - p_k''|| \leq ||q - q_k|| +||q_k - p_k|| +||p_k - p_k'|| + ||p_k' - p_k''|| 
%                  \leq \delta + 2d\delta + d\delta + 2d\delta 
%                  \leq 6d\delta = \eps/2.
%$
%
For $q = \point{\fq}{z}{u}$ and $q_k = \point{\fq}{h}{u}$, since $h\leq z \leq h+\delta$ then $||q - q_k|| \leq \delta$.    

Since $p_k = \rotate{q_j}{u_m}$, then invariant (I4) implies that $||u_m - u_k|| \leq (k-m)2\theta \leq d 2\theta$ and hence $||p_k - q_k|| \leq 2 \eps \sin((1/2) d 2\theta) \leq \eps 2(d \delta/\eps) \leq 2d \delta$.  

Finally, for $j=k$, Lemma~\ref{lem:key} implies that $||p_k' - p_k|| \leq k\delta\leq d\delta $ , that $||p'_k - p''_k|| \leq 2(k-m)\delta \leq 2d \delta$, and $p_k'' \in\conv{K_f}$. 

It follows that
\begin{eqnarray*}
|| q - p_k''|| & \leq & ||q - q_k|| +||q_k - p_k|| +||p_k - p_k'|| + ||p_k' - p_k''|| \\ 
                  & \leq & \delta + 2d\delta + d\delta + 2d\delta \\
                  & \leq & 6d\delta = \eps/2,
\end{eqnarray*}
as desired.
\end{proof}

%%%%%%%%%%%%%%%%%%%%%%%%%%%%%%%%%%%%%%%%%%%%%%
\subsection{Remarks}
\begin{itemize} \denselist
\item[(1)] For $d=2, 3$,  $\opt{P}{\eps/2}$ is only a factor of $O(1)$ and $O(\log(1/\eps))$, respectively, larger than $\opt{P}{\eps}$; therefore, the sizes of optimal $\eps$-kernels in these dimensions are relative stable. However, for $d \geq 4$, the stability drastically reduces in the worst case because of the superlinear dependency on $\opt{P}{\eps}$.

\item[(2)] Neither the upper nor the lower bound in the theorem is tight. For $d=3$, we can prove a tighter lower bound of $\Omega\bigl(\opt{P}{\eps} \log (1/ (\eps \cdot \opt{P}{\eps}))\bigr)$. We conjecture in $\reals^d$ that 
$$
\opt{P}{\eps/2} = \Theta\bigl(\opt{P}{\eps}^{\lfloor d/2 \rfloor} \log^{d-2} (1/(\eps^{(d-1)/2} \cdot \opt{P}{\eps})) \bigr).
$$ 
\end{itemize}

%%%%%%%%%%%%%%%%%%%%%%%%%%%%%%%%%%%%%%%%%%%%%%%%%%%
%%%%%%%%%%%%%%%%%%%%%%%%%%%%%%%%%%%%%%%%%%%%%%%%%%%
%%%%%%%%%%%%%%%%%%%%%%%%%%%%%%%%%%%%%%%%%%%%%%%%%%%
%%%%%%%%%%%%%%%%%%%%%%%%%%%%%%%%%%%%%%%%%%%%%%%%%%%

\bibliographystyle{mystyle}
\bibliography{bib-coreset}

%%%%%%%%%%%%%%%%%%%%%%%%%%%%%%%%%%%%%%%%%%%%%%%%%%%
%%%%%%%%%%%%%%%%%%%%%%%%%%%%%%%%%%%%%%%%%%%%%%%%%%%
%%%%%%%%%%%%%%%%%%%%%%%%%%%%%%%%%%%%%%%%%%%%%%%%%%%
%%%%%%%%%%%%%%%%%%%%%%%%%%%%%%%%%%%%%%%%%%%%%%%%%%%

\end{document}